\documentclass[10pt,conference,epsfig]{IEEEtran}
\IEEEoverridecommandlockouts
\usepackage{graphicx}
\usepackage{amsmath}
\usepackage{subfigure}
\usepackage{amssymb}
\usepackage{setspace}
\usepackage{stfloats}
\usepackage{amsthm}
\usepackage{subeqnarray}
\usepackage{cases}
\usepackage{multirow}
\usepackage{slashbox}
\usepackage{algorithm}
\usepackage{algpseudocode}
\usepackage{cite}
\usepackage{algorithm}
\usepackage{booktabs}
\usepackage{etoolbox}
\usepackage{color}
\usepackage{bm}
\usepackage{stfloats}
\usepackage{amsthm,amsmath,amssymb}
\usepackage{mathrsfs}
\usepackage{stfloats}
\usepackage{array}
\usepackage{caption}
\makeatletter
\patchcmd{\@makecaption}
  {\scshape}
  {}
  {}
  {}
\makeatletter
\patchcmd{\@makecaption}
  {\\}
  {:\ }
  {}
  {}
\makeatother

\newtheorem{remark}{\emph {Remark}}

\newtheorem{theorem}{\emph {Theorem}}
\makeatother
\begin{document}
\title{Improving THz Coverage for 6G URLLC Services via Exploiting Mobile Computing}
\author{
{Sha Xie, Lingxiang Li, Zhi Chen, and Shaoqian Li} \\
National Key Laboratory of Science and Technology on Communications, \\
University of Electronic Science and Technology of China, Chengdu 611731, China\\
%\IEEEauthorrefmark{1}The corresponding author, email: chenzhi@uestc.edu.cn
}

\maketitle \pagestyle{empty} \thispagestyle{empty}
\maketitle
\begin{abstract}
Terahertz (THz) communication (0.1-10 THz) is regarded as a promising technology, which provides rich available bandwidth and high data rate of terahertz bit per second (Tbps). However, THz signals suffer from high path loss, which profoundly decreases the transmission distance. To improve THz coverage, we consider the aid of mobile computing. Specifically, job offloading decision in mobile computing and frequency allocation in communication are co-designed to maximize distance and concurrently support ultra-reliable low-latency communications (URLLC) services for the sixth-generation (6G) mobile communication. Further, the above optimization problem is non-convex, then an effective and low-complexity method is proposed via exploiting the special structure of this problem. Finally, numerical results verify the effectiveness of our work.
%we focus on addressing the distance challenge in the THz,  the distance improvement is significant with our resource allocation

\end{abstract}


%\renewcommand{\thefootnote}{}
%\footnotetext{Corresponding Author: Zhi Chen, email: chenzhi@uestc.edu.cn.}
%\thanks{This work was supported by the National Natural Science Foundation of China under Grant 61631004.}

\section{Introduction}
With the development of smart terminals and emerging applications, wireless data traffic has increased dramatically. The fifth-generation (5G) mobile communication cannot meet the requirements of intelligent and automated systems 10 years later \cite{6G2}. As a result, the sixth-generation (6G) mobile communication has been attracting a lot of attentions. 6G is envisioned to utilize Terahertz (THz) band, in which the available bandwidth is theoretically three orders of magnitude wider than that in the millimeter wave (mmWave) band \cite{liang1}\cite{liang2}.
THz communication can provide high data rate of terahertz bit per second (Tbps), making that every device access high-speed applications smoothly \cite{THz1}\cite{THz3}. However, due to high path loss and molecular absorption attenuation of THz wave propagation, THz communication coverage is usually very limited, which poses the main challenges for THz communications \cite{han}.

Currently, there are many works aiming to improve THz communication coverage. For multi-antenna scenario, beamforming technique is considered to get high directional gains of antenna arrays. Due to the use of massive antennas, these kind of methods usually require high energy consumption and thus limiting the application of THz communications \cite{green1}\cite{green2}. For single-antenna scenario, the works \cite{han} and \cite{han1} tried to improve THz transmission distance by optimizing spectrum, transmit power, modulations and the number of frames. The work \cite{coverage1} derived new expressions for coverage probability of downlink transmissions, where the coverage probability is defined as the probability that SINR is larger than a predefined threshold. While the works \cite{han} and \cite{han1} captured the distance-frequency-dependent characteristics of THz channels, for which they maximized the achievable transmission distance of users subject to the constraints that the achievable communication data rate associated with users is no lower than a predefined common threshold.

However, existing works try to improve THz communication coverage without taking into account the attributes of applications as well as mobile devices' ability. In the future we will have a wide range of applications such as autonomous vehicles, AR/VR, facial recognition, etc \cite{6G3}. These emerging applications require ultra-reliable low-latency communications (URLLC) services and involve transmissions of long-packets. Research on URLLC is extensively done in 5G, but only for short-packet transmissions, where the transmission delay is simply approximated as one time slot \cite{10}. This approximation, however, is not suitable for long-packet transmissions anymore since long-packet transmissions cannot be completed in one time slot. As such, new metrics as well as math descriptions is needed to characterize the URLLC constraints for long-packet transmissions for applications such as VR/AR in the future 6G networks. On the other hand, with the development of semiconductor industry the mobile devices are becoming more and more powerful, enabling local computing of partial or whole jobs at terminals, which is also referred to as mobile computing. This mobile computing, if properly used, can help to improve the THz communication coverage. For example, when users offload all the jobs to the access point for processing, the communication burden is heavy, resulting in limited THz communication coverage. By comparison, if users can process part of these jobs, the communication burden is relieved and the communication data rate required is also reduced, which in turn helps to improve the THz communication coverage.

In this paper, we consider the same multi-user scenario as in \cite{han}. In contrast to the coarse-grained resource plans where each user is allocated with a predefined common communication data rate, a more fine-grained resource scheme is introduced to make full use of the abilities across all users, not only the communication ability but also the computing ability. To the best of our knowledge, this is the first time to improve THz coverage via exploiting mobile computing.
%To that end, there are some technical challenges. First, how to accurately characterize the URLLC requirements for long-packet transmissions? Second, how to exploit the distance-frequency-dependent characteristics of THz channels for system design? Third, how to collaboratively design communication and computing?
Our main contributions can be summarized as follows.
\begin{itemize}
  \item We investigate the relationship between the communication data rate, the transmission distance and the carrier frequency in the THz Band, obtaining in closed-form the expression of transmission distance with respect to the carrier frequency and the rate via fitting method.
  \item We develop a novel reliability framework to characterize URLLC constraints for long-packet transmissions, where the reliability is defined as the probability that end-to-end (E2E) delay remains below a certain predefined threshold. The E2E delay equals a weighted sum of local computing delay, communication delay and edge computing delay.
  \item With the proposed framework for long-packet transmissions with URLLC constraints, we formulate an achievable sum communication distance maximization problem with respect to the optimization of carrier frequency allocations and job offloading decisions across all users, thus improving the THz communication coverage.
  \item The achievable sum communication distance maximization problem is non-convex. To solve that issue, we first do problem reformulation via exploiting the special structure of the problem. An effective and low-complexity method is then proposed by decomposing the original problem into two sub-problems, i.e., the data rate threshold minimization problem and the sum communication distance maximization problem.
\end{itemize}

\vspace{-2pt}
\section{System Model and Problem Formulation}
\vspace{-3pt}
We consider $K$ mobile users and an edge server in a mobile edge system (see Fig.\ref{system}). On the one hand, when the mobile user does not have any local computation capabilities, all tasks must be offloaded to the edge server through a wireless communication link. This undoubtedly increases transmission overhead, and brings a short communication distance and a high transmission delay. On the other hand, with the development of semiconductor industry the mobile devices are becoming more and more powerful, enabling local
computing of partial or whole jobs at terminals, which can help to improve the THz communication coverage. In light of this, we consider each user has local computation capability.

\vspace{-5pt}
\begin{figure}[H]
\begin{center}
            \includegraphics[scale=0.415]{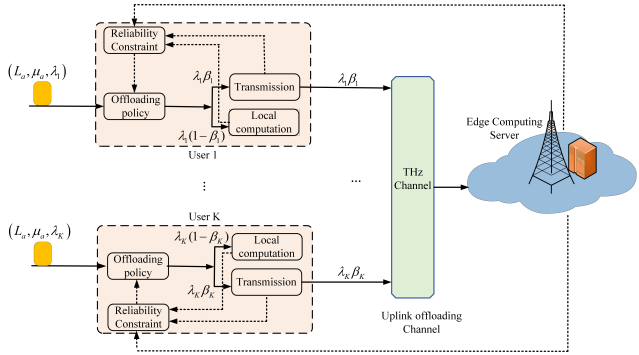}
           \caption{A multi-user mobile computing system.}
            \label{system}
            \end{center}
        \end{figure}
\vspace{-18pt}
For simplicity, each computation task is characterized by a triple of parameters $< {L_a},{\mu _a},\lambda_k > ,k = 1,2,...,K$, where ${L_a}$ (in bits) is the average input data size per job, ${\mu _a}$ (in cycles) denotes the average CPU cycles needed by completing a job, and $\lambda_k$ (in jobs/s) represents the job arrival rate at the $k$-user. Both the input data size and the CPU cycles required to complete a job follow exponential distributions.

In the following, the transmission data in the THz communication model is first introduced in Sec.II-A, then the reliability framework to characterize URLLC constraints is given in the computing model (see Sec.II-B). Finally, in Sec.II-C we present a distance maximization problem by jointly designing offloading decision and frequency allocation under 6G URLLC, computation and communication constraints.

\vspace{-9pt}
\subsection{THz Communication Model}
\vspace{-6pt}
In the THz Band, the transmitted signals experience severe molecular absorption loss and free-space loss \cite{channel}. Specifically, based on Recommendation ITU-R P.676-9 \cite{ITU-R}, the specific gaseous attenuation is given by
\vspace{-5pt}
\begin{align}
\gamma (\tilde f){\rm{ = }} 0.1820 \tilde f \cdot {N^{''}}(\tilde f)\quad \rm{dB/km{\rm{}}},\label{equ:ITU}
\end{align}
\vspace{-1.5pt}
\!\!where $\tilde f$ represents the frequency, and ${N^{''}}(\tilde f)$ is the imaginary part of the frequency-dependent complex refractivity. In consequence, combined with communication distance $d$, the molecular absorption loss $L_{abs}{(\tilde f,d)_{[dB]}}$ in dB in the THz Band can be described as the product of $\gamma (\tilde f)$ and $d$:
\vspace{-5pt}
\begin{align}
L_{abs}{(\tilde f,d)_{[dB]}} = \gamma (\tilde f) \cdot d = 0.1820 \tilde f \cdot {N^{''}}(\tilde f) \cdot d{\rm{ }}.\label{equ:fit}
\end{align}
\vspace{-1.5pt}
Additionally, the free-space loss can be written in dB as
\vspace{-5pt}
\begin{align}
{L_{spread}}{(\tilde f,d)_{[dB]}} = 20{\log _{10}}(\tfrac{{4\pi \tilde fd}}{c}),
\end{align}
\vspace{-1.5pt}
\!\!where $c = 3 \times {10^8}$ m/s stands for the speed of light. Therefore, the total path loss ${L}{(\tilde f,d)_{[dB]}}$ can be given by
\begin{align}
{L}{(\tilde f,d)_{[dB]}} &= {L_{abs}}{(\tilde f,d)_{[dB]}} + {L_{spread}}{(\tilde f,d)_{[dB]}} \nonumber\\
%&= \gamma (f) \cdot d{\rm{ + }}20{\log _{10}}(\tfrac{{4\pi fd}}{c})\nonumber\\
                     &= 0.1820{N^{''}}(\tilde f) \cdot \tilde f \cdot d{\rm{ + }}20{\log _{10}}(\tfrac{{4\pi \tilde fd}}{c}).
                     \label{equ:L}
\end{align}
The achievable data rate ${R_k}$ of the $k$-th user is described as
\begin{align}\label{equ:rate}
\quad{R_k}&= {B_k}\cdot {\log _2}\left( {1 + \tfrac{{{p_k}{G_t}{G_r}}}{{{\sigma ^2} \cdot {{10}^{\frac{{{L}{{({\tilde f_k},{d_k})}_{[dB]}}}}{{10}}}}}}} \right),
\end{align}
where $B_k$ represents the available communication bandwidth of the $k$-th user, ${G_t}$/${G_r}$ stands for the transmitting/receiving antenna gain, and ${\sigma ^2}$ refers to the noise power.
\vspace{-9pt}
\subsection{Computing Model}
\vspace{-6pt}
We consider the same computing model as in our previous work \cite{precious}. As each mobile user has local computation power, user can choose to compute its job locally or offload this job to edge server. Therefore, computing model consists of local computing model and edge computing model. Moveover, in the local computing model there is only one phase, i.e., local job execution, while the edge computing model includes two phases, i.e., job offloading and edge job execution. Therefore, the system-level reliability framework should be developed based on these three phases. In the following, we introduce the reliability framework.
%In the sequel, we only describe the main results for the computing model.
For more details, please refer to \cite{precious}.
\subsubsection{Reliability of the system ${\Phi _k}(\varepsilon )$}

\

\noindent  \ \
We define the reliability of the system as the probability that the total delay $T_k$ remains below a certain pre-defined threshold $\varepsilon $ for the $k$-th user.
Here, a small threshold $\varepsilon $ meets the low-latency requirement while a high probability value of $\mbox{Pr}\{T_k \le \varepsilon\}$ meets the
high-reliability requirement. Let ${\beta _k}$ be the offloading probability of user $k$, local computing probability thus is $1-{\beta _k}$. Taking into account the binary computing choices,
the reliability of user $k$ is given by
\begin{equation}
{\Phi _k}(\varepsilon )\triangleq \mbox{Pr}\{T_k \le \varepsilon\} \label{eqReliability}
%&= \mbox{Pr}\{A_k =0\}\mbox{Pr}{\rm{\{T}}_{l,k} \le \varepsilon {\rm{\}}} + \mbox{Pr}\{A_k =1\}\mbox{Pr}{\rm{\{T}}_{m,k} \le \varepsilon {\rm{\}}}\nonumber\\
%&= \mbox{Pr}\{A_k =0\}\Phi _{l,k}(\varepsilon ) + \mbox{Pr}\{A_k =1\}\Phi _{m,k}(\varepsilon )\nonumber\\
={\rm{(1}} - {\beta _k}{\rm{)}}\Phi _{l,k}(\varepsilon ) + {\beta _k}\Phi _{m,k}(\varepsilon ),
\end{equation}
where $\Phi _{l,k}(\varepsilon) $ and $\Phi _{m,k}(\varepsilon)$ refer to the reliabilities of local computing and edge computing, respectively. In what follows, we discuss $\Phi _{l,k}(\varepsilon) $ and $\Phi _{m,k}(\varepsilon)$, respectively.

\subsubsection{Reliability of the local computing $\Phi _{l,k}(\varepsilon) $}

\

\noindent  \ \
The service time of user $k$ follows an exponential distribution with rate parameter $\mu _{l,k}={{f_{l,k}}}/{{{\mu _a}}}$, where $f_{l,k}$ is the CPU-cycle frequency of the $k$-th user.
Therefore, we obtain
\vspace{-4pt}
\begin{equation}
\Phi _{l,k}(\varepsilon) =\mbox{Pr}{\rm{(T}}_{l,k} \le \varepsilon)=1 - {e^{ - {\rm{(}}\mu _{l,k} - (1 - {\beta _k}){\lambda _k}{\rm{)}}\varepsilon}},
\end{equation}
\vspace{-2pt}
\!\!where ${\rm{T}}_{l,k}$ denotes the total delay for local computing and $(1 - {\beta _k}){\lambda _k} < \mu _{l,k}$.

\subsubsection{Reliability of the edge computing $\Phi _{m,k}(\varepsilon)$}
\

\noindent  \ \
Let ${\rm{T}}_{m,k}$, $t_1$ and $t_2$ denote the total delay, the transmission delay and the job execution time for edge computing, respectively. We arrive at
\vspace{-4pt}
\begin{equation}
\Phi _{m,k}(\varepsilon ) = \mbox{Pr}{\rm{\{T}}_{m,k} \le \varepsilon {\rm{\}}}
= \mbox{Pr}\{t_1 \le \varepsilon\} * \mbox{Pr}\{t_2 = \varepsilon\},
\end{equation}
\vspace{-0.5pt}
\!\!where $*$ represents convolution operator.
 $\mbox{Pr}\{t_1 \le \varepsilon\}=1 - {e^{ - {\rm{(}}\frac{{{R_k}}}{{{L_a}}} - {\beta _k}{\lambda _k}{\rm{)}}\varepsilon}}$,
where ${\beta _k}{\lambda _k} < {{{R_k}}}/{{{L_a}}}$.
In addition, $\mbox{Pr}\{t_2 = \varepsilon\} = ( \mu _m - {\beta _k}{\lambda _k}{\rm{)}}{e^{ - {\rm{(}}\mu_m - {\beta _k}{\lambda _k}{\rm{)}}\varepsilon}}$,
where $\mu _m={{{f_m}}}/{{{\mu _a}}}$ is the service rate of the edge server, $f_m$ is the CPU-cycle frequency of the edge server, and $\sum\nolimits_{k = 1}^K {{\beta _k}{\lambda _k}}  < {\mu _m}$.

\vspace{-4pt}
\subsection{Problem Formulation}
\vspace{-5pt}
%As mentioned above, THz signals suffer from high attenuation, which results in a limited communication distance. To tackle this,
In this paper, we aim to enhance THz coverage aided by mobile computing. Let us focus on the relationship between communication and computing. On the one hand, quality of wireless transmission link as well as local/edge computation power has a significant impact on task offloading.
On the other hand, communication resource (e.g., frequency, available bandwidth and transmission power) and offloading policy affect achievable communication distance. Motivated by these observations, we consider the joint design of communication and computing to maximize coverage while providing 6G URLLC services. Besides, due to the very high frequency-selectivity of the THz channel, from the perspective of communication, we mainly focus on the impact of frequency on the joint design. Thus, we assume each user's bandwidth and transmission power are the same, while job arrival rate and computation capacity are different. Then, offloading policy and frequency allocation for multiple users are formulated as an optimization problem, where the objective is to maximize distance under the constraints of URLLC, computation and frequency resources.
Specifically, the optimization problem is:
\vspace{-14pt}
\begin{subequations}
\begin{align}
\mbox{(P1)}\,&\max\limits_{{\bf{d}},{\boldsymbol{\beta} },{\bf{\tilde f}}} {\rm{ }}\sum\nolimits_{k = 1}^K {{d_k}}   \nonumber \\
&\ {{\rm{s.t.}}}   \, \    {\rm{        (1}} - {\beta _k}{\rm{)}}{\Phi _{l,k}}(\varepsilon) + {\beta _k}{\Phi _{m,k}}(\varepsilon) \ge {\vartheta _{th}}, \forall k, \label{equ:1a} \\
&\qquad\ {\rm{        (1}} - {\beta _k}{\rm{)}}{\lambda _k} < {\mu _{l,k}}{\rm{, }}\forall k ,\label{equ:1b}\\
&\qquad\ {\rm{        }}{\beta _k}{\lambda _k} < \tfrac{{{R_k}}}{{{L_a}}}{\rm{, }}\forall k,\label{equ:1c} \\
&\qquad\ {\rm{        }}\sum\nolimits_{k = 1}^K {{\beta _k}{\lambda _k}}  < {\mu _m}{\rm{, }}\label{equ:1d}\\
&\qquad\ {\rm{        0}} \le {\beta _k} \le {\rm{1, }}\forall k,\label{equ:1e}\\
               & \qquad\ {\tilde f_k} \in {{\bf{\tilde f}}_0},{\tilde f_k} \ne {\tilde f_j},k \ne j,\forall k,j \label{equ:1f},
\end{align}
\end{subequations}
where ${\boldsymbol{\beta} } \!\!\!= \!\!\!\{{\beta _1},...{\beta _k},...,{\beta _K}\!\}$, ${\bf{\tilde f}}\!\!\! =\!\!\! \{{\tilde f_1},...{\tilde f_k},...,{\tilde f_K}\!\}$ and ${\bf{d}} \!\!\!= \!\!\!\{{d_1},...{d_k},...,{d_K}\!\}$ denote the sets of offloading policies, frequencies, and achievable distances, respectively. ${\vartheta _{th}}$ is a certain pre-defined reliability threshold, and ${{\bf{\tilde f}}_0}$ indicates the set of the $K$ valid frequencies with ${{\bf{\tilde f}}_0} \!\!=\! \!\{ {\tilde f_0^1,...\tilde f_0^k,...,\tilde f_0^K} \!\}$. (\ref{equ:1a}) represents the reliability constraint. (\ref{equ:1b}), (\ref{equ:1c}), and (\ref{equ:1d}) guarantee that the three queues in local computing, job offloading and job execution are stable, respectively. Finally, (\ref{equ:1e}) defines the feasible region of the offloading probability, and (\ref{equ:1f}) shows that the frequency occupied by each user is different.

It can be noted that the solution of this optimization problem contains the following three difficulties:
\vspace{-2pt}
\begin{enumerate}
  \item As ${N^{''}}(\tilde f)$ in (\ref{equ:fit}) is based on the spectroscopic data in Table 1 and Table 2 of \cite{ITU-R}, the closed-form expression of $\gamma $ with respect to $\tilde f$ cannot be obtained easily.
  \item The reliability constraint includes convolution operator, which makes the problem solution more difficult.
  \item This problem is non-convex, which cannot be directly solved by standard convex optimization techniques.
\end{enumerate}
\vspace{-2pt}

In the following, we first handle the first two issues, then propose an effective and low-complexity method to solve (P1).

\vspace{-5pt}
\section{Multi-User Computing Offloading Policies and Communication Frequency Allocation}
\vspace{-4pt}
In this section, we aim to address the above three difficulties and derive the optimal offloading decisions and frequency allocation. Towards these goals, in Sec.III-A, we obtain the closed-form expression of total path loss with respect to frequency via fitting method, which addresses the first difficulty. Based on this, we obtain the rate expression corresponding to the fitted total path loss. Next, in Sec.III-B, we reformulate the problem by completing convolution operation, analyzing the relationships between data rate constraints, and deriving the closed-form distance expression, which addresses the second difficulty. Finally, in Sec.III-C we propose an effective and low-complexity method to obtain the optimal solutions.
\vspace{-5pt}
\subsection{Total path loss fitting function}
\vspace{-4pt}
According to (\ref{equ:rate}), frequency affects communication rate, thus further affects achievable communication distance. Therefore, to improve coverage in terms of frequency allocation, it is necessary to find out the relationships between data rate, frequency and distance. Owing to the first difficulty, we consider to fit $\gamma $, thereby obtaining approximate closed-form expression of total path loss. Due to limitations of space, we only show the fitted specific gaseous attenuation $\gamma'$, which is described as the superposition of multiple Gaussian functions:
\begin{align}
\gamma'(\tilde f) = \sum\nolimits_{i = 1}^7 {{a_i}{e^{ - {{(\frac{{{\tilde f} - {b_i}}}{{{c_i}}})}^2}}}},
\label{equ:fitted}
\end{align}
where $100 \le \tilde f \le 1000$ (in units GHz). All these parameters are shown in Table I.

Substituting (\ref{equ:fitted}) into (\ref{equ:L}), we derive
\begin{align}
{L'}{(\tilde f,d)_{[dB]}} %&= \gamma' (f) \cdot d{\rm{ + }}20{\log _{10}}(\frac{{4\pi fd}}{c}){\rm{                    }}\nonumber\\
                     &=  \sum\nolimits_{i = 1}^7 {{a_i}{e^{ - {{(\frac{{{\tilde f} - {b_i}}}{{{c_i}}})}^2}}}} \cdot d{\rm{ + }}20{\log _{10}}(\tfrac{{4\pi \tilde fd}}{c}).\!
\end{align}
Therefore, the data rate in (\ref{equ:rate}) can be further written as
\begin{align}\label{equ:rate2}
R_k={B_k}\cdot{\log _2}\left( {1 + \tfrac{{{p_k}{G_t}{G_r}}}{{{\sigma ^2} \cdot {{10}^{\frac{{{\gamma '}(\tilde f_k) \cdot {d_k}{\rm{ + }}20{{\log }_{10}}(\frac{{4\pi {\tilde f_k}{d_k}}}{c}){\rm{ }}}}{{10}}}}}}} \right).\!
\end{align}

{\begin{remark}
This fitting method can be applied to another THz path loss expression \cite{channel},
%\begin{align}
${L}(\tilde f,d)={e^{{\kappa _{abs}}(\tilde f)d}}  {{\rm{(}}{{4\pi \tilde fd}}/{c}{\rm{)}}^2}$,
%\end{align}
where ${\kappa _{abs}}$ is the medium absorption coefficient.
%As some parameters used for computing ${\kappa _{abs}}$ are directly read from the HITRAN database as mentioned in \cite{channel}, it is difficult to derive the expression for ${\kappa _{abs}}$ with respect to $f$. As a result,
%Specifically, fitting the molecular absorption loss with respect to $f$, the total path loss as a function of $f$ and $d$ can also be obtained.
\end{remark}}

\vspace{-15pt}
\begin{table}[H]
\setlength{\abovecaptionskip}{-0.05cm}
\caption{All Parameter Settings in $\gamma'$}
\small
\centering
\setlength{\tabcolsep}{0.15mm}{
\begin{tabular}{m{1cm}<{\centering}m{1.6cm}<{\centering}|m{1.1cm}<{\centering}m{1.5cm}<{\centering}|m{1.1cm}<{\centering}m{1.3cm}<{\centering}}
%{cc|cc|cc}
\toprule
Parameter    & Value                       & Parameter    & Value                          & Parameter    & Value \\[-2.5pt]
\midrule
$a_1$        & $9906$                      & $b_1$        & $557$                          & $c_1$        & $3.175$ \\
$a_2$        & $9940$                      & $b_2$        & $752.1$                        & $c_2$        & $4.968$ \\
$a_3$        & $7301$                      & $b_3$        & $987.9$                        & $c_3$        & $4.6$   \\
$a_4$        & $5667$                      & $b_4$        & $556.5$                        & $c_4$        & $8.772$ \\
$a_5$        & $542.2$                     & $b_5$        & $559.1$                        & $c_5$        & $33.58$ \\
$a_6$        & $3.338 \times {10^{15}}$    & $b_6$        & $1.46 \times {10^{4}}$         & $c_6$        & $2496$  \\
$a_7$        & $208.2$                     & $b_7$        & $447.7$                        & $c_7$        & $6.968$ \\
\bottomrule
\end{tabular}}
\end{table}
\vspace{-15pt}

\subsection{Problem Reformulation}
In this subsection, we aim to reformulate (P1). We first complete the convolution operation in (\ref{equ:1a}) to derive the corresponding relationship between $R_k$ and $\beta_k$ for the $k$-th user, $\forall k$. Then combined with (\ref{equ:1c}), we can simplify (P1).
Specifically, we let ${u_k} = \frac{{{R_k}}}{{{L_a}}} - {\beta _k}{\lambda _k} > 0$, ${v_k} = {\mu _m} - {\beta _k}{\lambda _k} > 0$, $q_k(\varepsilon ) = 1 - {e^{ - {u_k}\varepsilon }}$ and $g_k(\varepsilon ) = {v_k}{e^{ - {v_k}\varepsilon }}$. As the delay threshold is positive, $q_k(\varepsilon )$ and $g_k(\varepsilon )$ can be further written as $q_k(\varepsilon ) = (1 - {e^{ - {u_k}\varepsilon }}) \cdot h(\varepsilon )$ and $g_k(\varepsilon ) = {v_k}{e^{ - {v_k}\varepsilon }}\cdot h(\varepsilon )$, respectively, where $h(\varepsilon )$ is unit step function. Then we obtain
\vspace{-0.3cm}
\begin{align}\label{equ:juanji}
{\Phi _{m,k}}(\varepsilon ) &= q_k(\varepsilon ) * g_k(\varepsilon ) \notag\\
   %             &= \int_{ - \infty }^{ + \infty } {f(\tau )g(\varepsilon  - \tau )d\tau }  \notag\\
                &= \int_0^\varepsilon  {(1 - {e^{ - {u_k}\tau }}){v_k}{e^{ - {v_k}(\varepsilon  - \tau )}}d\tau } \notag \notag\\
             \!  &\!\!\!\!\!\!\mathop {\rm{ = }}\limits^{({u_k} \ne {v_k})} {\rm{1}} - {e^{ - {v_k} \varepsilon}} + \frac{{v_k}}{{{v_k} - {u_k}}}({e^{ - {v_k}\varepsilon}} \!-\! {e^{ - {u_k}\varepsilon}}).\!
\end{align}

By far, we have finished convolution operation. Then, we simplify (\ref{equ:1a}).
Substituting (\ref{equ:juanji}) into (\ref{equ:1a}), we arrive at
\begin{align}\label{equ:reliability}
%{e^{ - b\varepsilon }} - \frac{b}{{b - a}}({e^{ - b\varepsilon }} - {e^{ - a\varepsilon }}) \le \frac{{(1 - {\vartheta _{th}}) - {\rm{(1}} - \beta {\rm{)}}{\Phi _l}(\beta )}}{\beta }\\
{\rm{                         }}\frac{{{e^{({v_k} \!-\! {u_k})\varepsilon }}\! \!-\!\! 1}}{{{v_k} \!-\! {u_k}}}{\rm{ }}\! \le \! \!- \frac{{{e^{{v_k}\varepsilon }}}}{{v_k}}(\!\frac{{{\vartheta _{th}} \!- \! {\rm{(1}} \!\!-\!\! \beta_k {\rm{)}}{\Phi _{l,k}}(\varepsilon)}}{\beta_k } + {e^{ - {v_k}\varepsilon }} \!- \! 1).\!\!
\end{align}
This, combined with $\Lambda_k \!=  - \frac{{{e^{{v_k}\varepsilon }}}}{{v_k}}(\frac{{{\vartheta _{th}}\! - {\rm{(1}} \!- \!\beta_k {\rm{)}}{\Phi _{l,k}}(\varepsilon)}}{\beta_k } + {e^{ - {v_k}\varepsilon }} \!-\! 1)$, $x_k = {v_k} - {u_k} < 0$, indicates that
\begin{align}
%\varpi x + 1 &\le {e^{\frac{{\varepsilon (\varpi x + 1)}}{\varpi} - \frac{\varepsilon }{\varpi}}}\notag\\
- \frac{\varepsilon }{\Lambda_k}(\Lambda_k x_k + 1){e^{ - (\frac{\varepsilon }{\Lambda_k}(\Lambda_k x_k + 1))}} &\ge  - \frac{\varepsilon }{\Lambda_k}{e^{ - \frac{\varepsilon }{\Lambda_k}}}\notag\\
  \Rightarrow - \frac{\varepsilon }{\Lambda_k}(\Lambda_k x_k + 1)& \ge W( - \frac{\varepsilon }{\Lambda_k}{e^{ - \frac{\varepsilon }{\Lambda_k}}})\notag\\
% x & \le  - \frac{{\frac{\varpi}{\varepsilon }W( - \frac{\varepsilon }{\varpi}{e^{ - \frac{\varepsilon }{\varpi}}}) + 1}}{\varpi}\notag\\
 \Rightarrow {R_k} & \ge {R_{k,th}},
\label{equ:new}
\end{align}
where $W( \cdot )$ refers to Lambert function, and
\begin{equation}
{R_{k,th}} = \left( {{\mu _m} + \tfrac{{\frac{\Lambda_k }{\varepsilon }W( - \frac{\varepsilon }{\Lambda_k }{e^{ - \frac{\varepsilon }{\Lambda_k }}}) + 1}}{\Lambda_k }} \right){L_a}.
\end{equation}

For ease of exposition, we called ${R_{k,th}}$ as the data rate threshold of the $k$-user.
Now, combined with constraints (\ref{equ:new}) and (\ref{equ:1c}), we further obtain the range of ${R_k}$ as follows,
\begin{align}
\left\{ {\begin{array}{*{20}{c}}
{\!\!\!\!\!\!{R_k} \ge {R_{k,th}}}\\
{{R_k} > {\beta _k}{\lambda _k}{L_a}}
\end{array}} \right. \forall k.
\label{equ:compare}
\end{align}

Let us focus on ${R_{k,th}}$ and ${\beta _k}{\lambda _k}{L_a}$, which can help simplify (\ref{equ:compare}). In general, ${R_{k,th}}>{\beta _k}{\lambda _k}{L_a}$ can always be satisfied considering all the parameter values in our communication model and computing model. As a result, (\ref{equ:1c}) can be omitted.

Moreover, according to the instantaneous data rate $R_k$ in (\ref{equ:rate2}), the distance can be calculated as follows,
\begin{align}
{d_k}{\rm{ = }}\frac{{20}}{{\gamma'({\tilde f_k}) \!\cdot\! \ln 10}}W\left\{ {\frac{{\gamma'({\tilde f_k}) \!\cdot\! \ln 10}}{{20{\tilde f_k}}}\! \cdot\! {{10}^{\frac{{\chi ({R_k})}}{{20}}}}} \right\},
\label{equ:fenxi1}
\end{align}
in which $\chi ({R_k}){\rm{ = 10lo}}{{\rm{g}}_{10}}\tfrac{{{p_k}{G_t}{G_r}}}{{({2^{\frac{{{R_k}}}{{{B_k}}}}} \!- \!1){\sigma ^2}}} -\! {\rm{20lo}}{{\rm{g}}_{10}}(\frac{{4\pi }}{c}){\rm{    }}$.
Therefore, the optimization problem can be reformulated as:
\begin{subequations}
\begin{align}
\mbox{(P2)}\,&\max\limits_{{\boldsymbol{\beta} },{\bf{\tilde f}},{\bf{R}}} {\rm{ }}\sum\nolimits_{k = 1}^K {\frac{{20}}{{\gamma'({\tilde f_k}) \!\cdot\! \ln 10}}W\left\{ {\frac{{\gamma'({\tilde f_k})\! \cdot \! \ln 10}}{{20{\tilde f_k}}} \!\cdot \!{{10}^{\frac{{\chi ({R_k})}}{{20}}}}} \right\}}   \nonumber \\
&\ {{\rm{s.t.}}}   \, \    {R_k} \ge {R_{k,th}},\forall k  ,\label{equ:3a}\\
&\qquad\  \rm(\ref{equ:1b}), \rm(\ref{equ:1d}), \rm(\ref{equ:1e}) \ \rm{and} \ \rm(\ref{equ:1f}),\nonumber
%&\qquad\ {\rm{        (1}} - {\beta _k}{\rm{)}}{\lambda _k} < {\mu _{l,k}}{\rm{, }}\forall k ,\label{equ:3b}\\
%&\qquad\ {\rm{        }}\sum\nolimits_{k = 1}^K {{\beta _k}{\lambda _k}}  < {\mu _m}{\rm{, }}\label{equ:3d}\\
%&\qquad\ {\rm{        0}} \le {\beta _k} \le {\rm{1, }}\forall k ,\label{equ:2e}\\
%               & \qquad\  {f_k} \in {{\bf{f}}_0},{f_k} \ne {f_j},k \ne j,\forall k,j \label{equ:3f},
\end{align}
\end{subequations}
where ${\bf{R}} = ({R_1},...{R_k},...,{R_K})$ denotes the set of data rates.

\subsection{Problem Solution}

%\vspace{-5cm}
\begin{figure*}[bp]
\setlength{\abovecaptionskip}{-0.5pt}
\centering
\begin{tabular}{cc}
\begin{minipage}[t]{0.334\linewidth}
    \includegraphics[width = 1\linewidth]{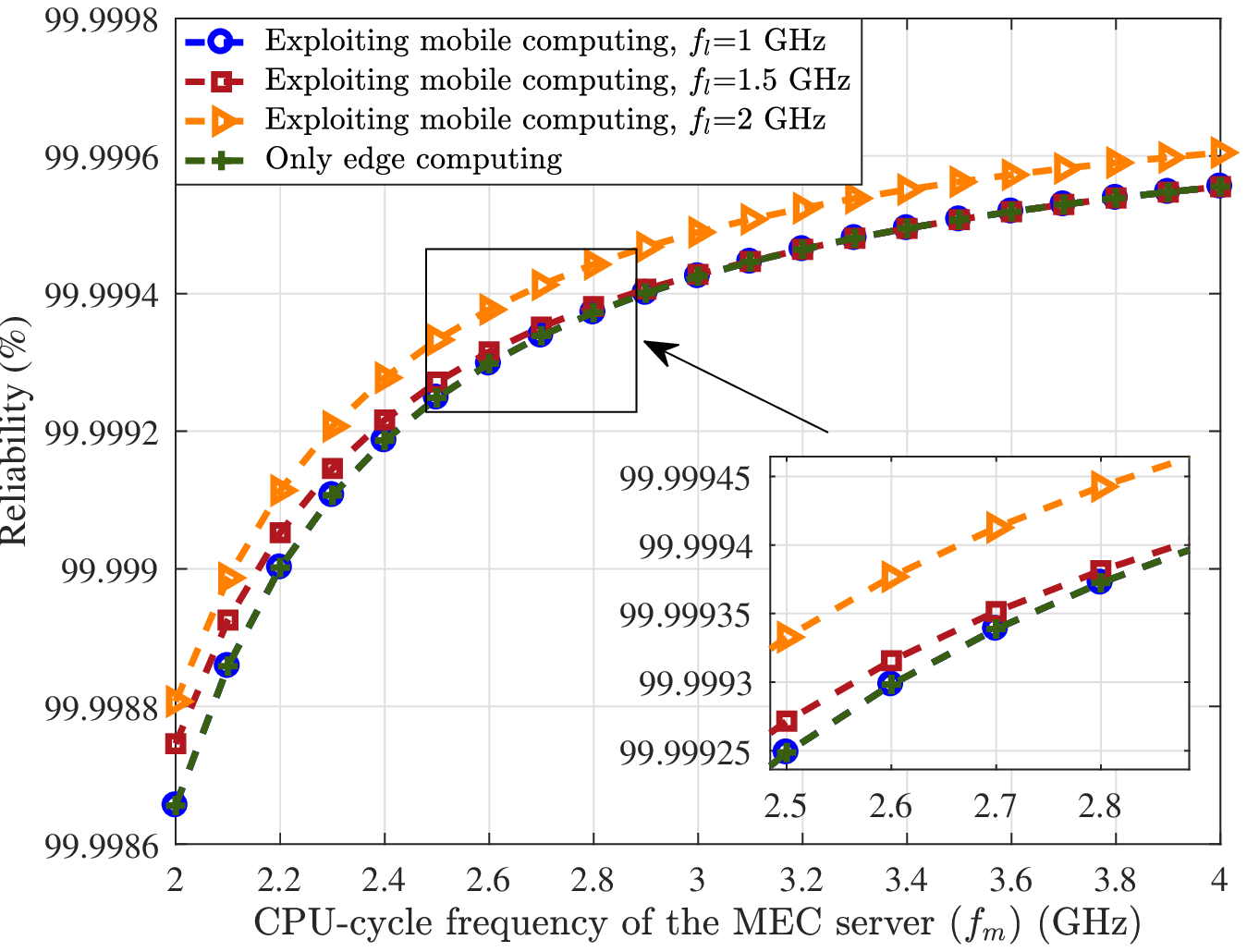}
           \caption{\scriptsize{Reliability in high-frequency communication.}}
            \label{reliability_THz}
\end{minipage}
\begin{minipage}[t]{0.334\linewidth}
    \includegraphics[width = 1\linewidth]{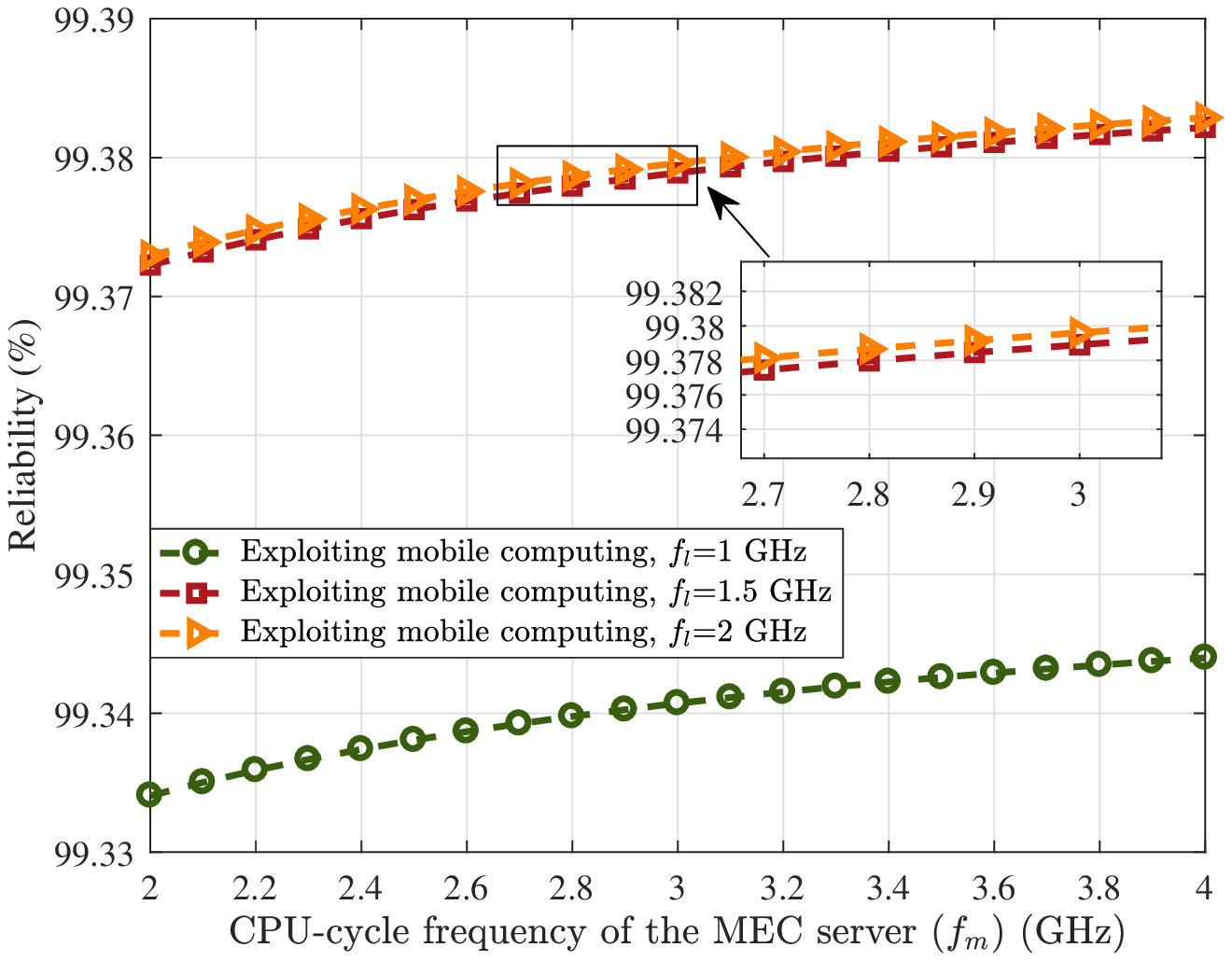}
           \caption{\scriptsize{Reliability in low-frequency communication.}}
            \label{reliability_4G_1}
\end{minipage}

\begin{minipage}[t]{0.333\linewidth}
    \includegraphics[width = 1\linewidth]{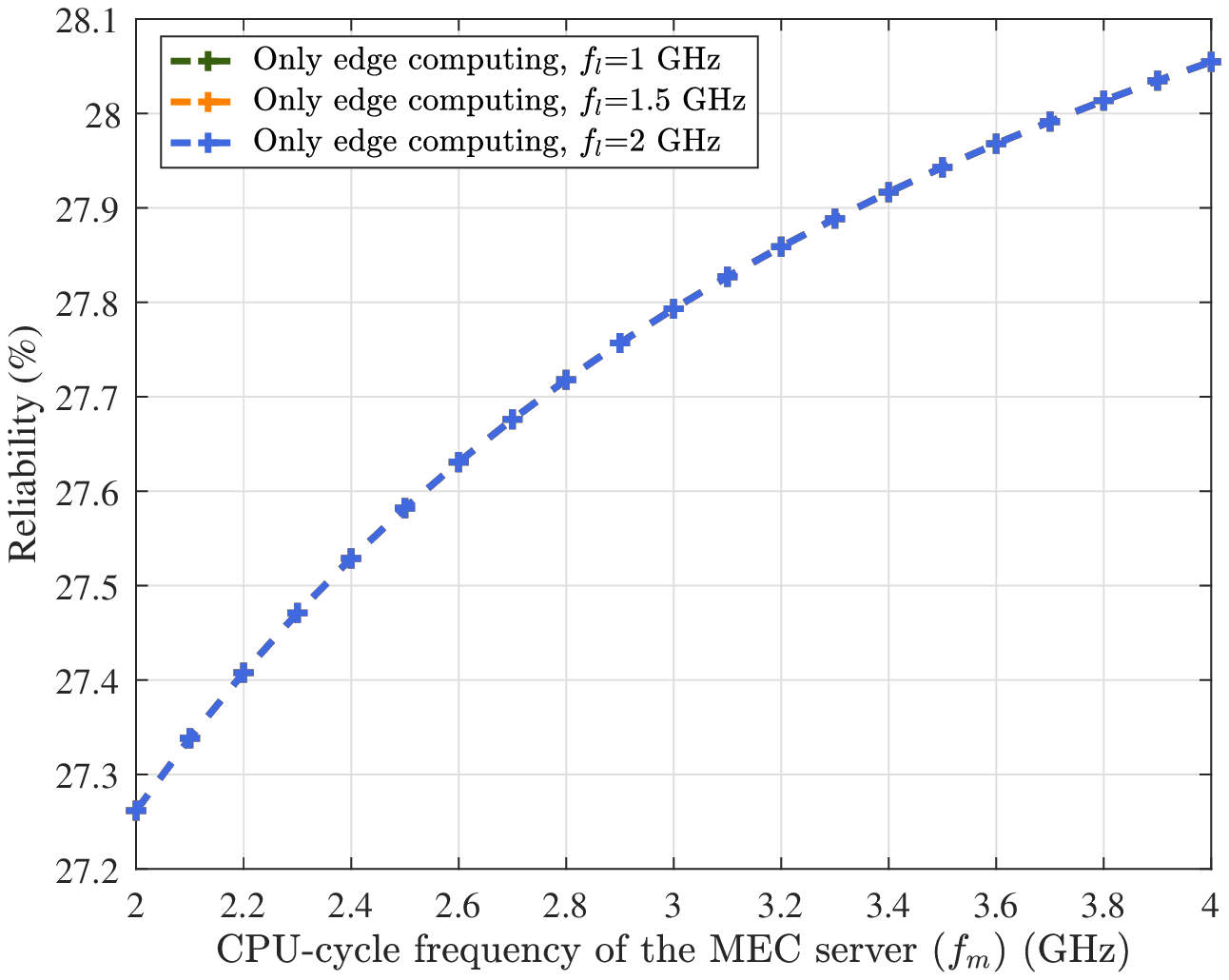}
           \caption{\scriptsize{Reliability in low-frequency communication.}}
            \label{reliability_4G_2}
\end{minipage}
\end{tabular}
\centering
\end{figure*}

\begin{figure*}[bp]
\setlength{\abovecaptionskip}{-0.5pt}
\begin{tabular}{cc}
\begin{minipage}[t]{0.303\linewidth}
    \includegraphics[width = 1\linewidth]{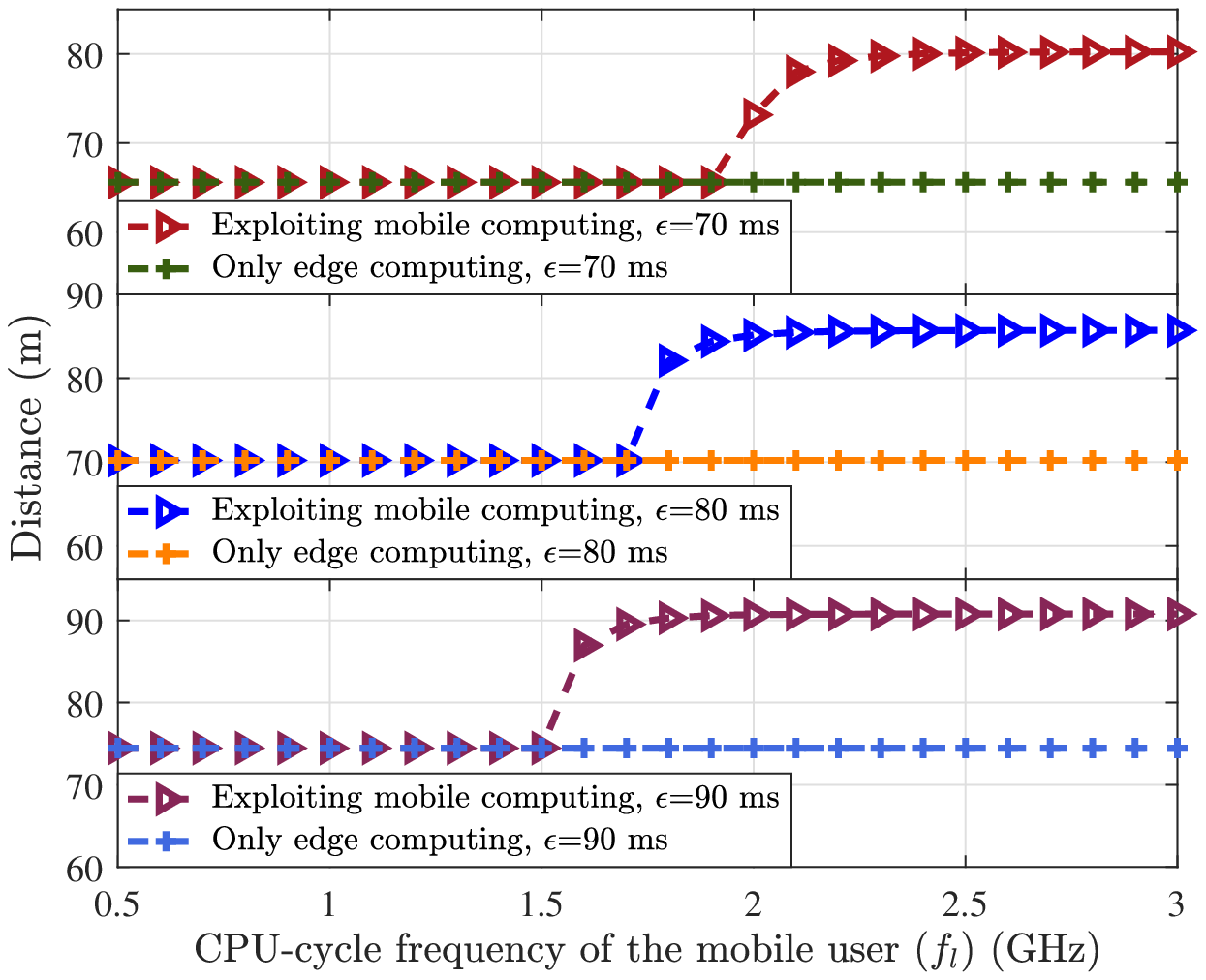}
\caption{\scriptsize{Communication coverage comparison in two different strategies.}}
            \label{dis_fl_epsilon}
\end{minipage}
\begin{minipage}[t]{0.205\linewidth}
    \includegraphics[width = 1\linewidth]{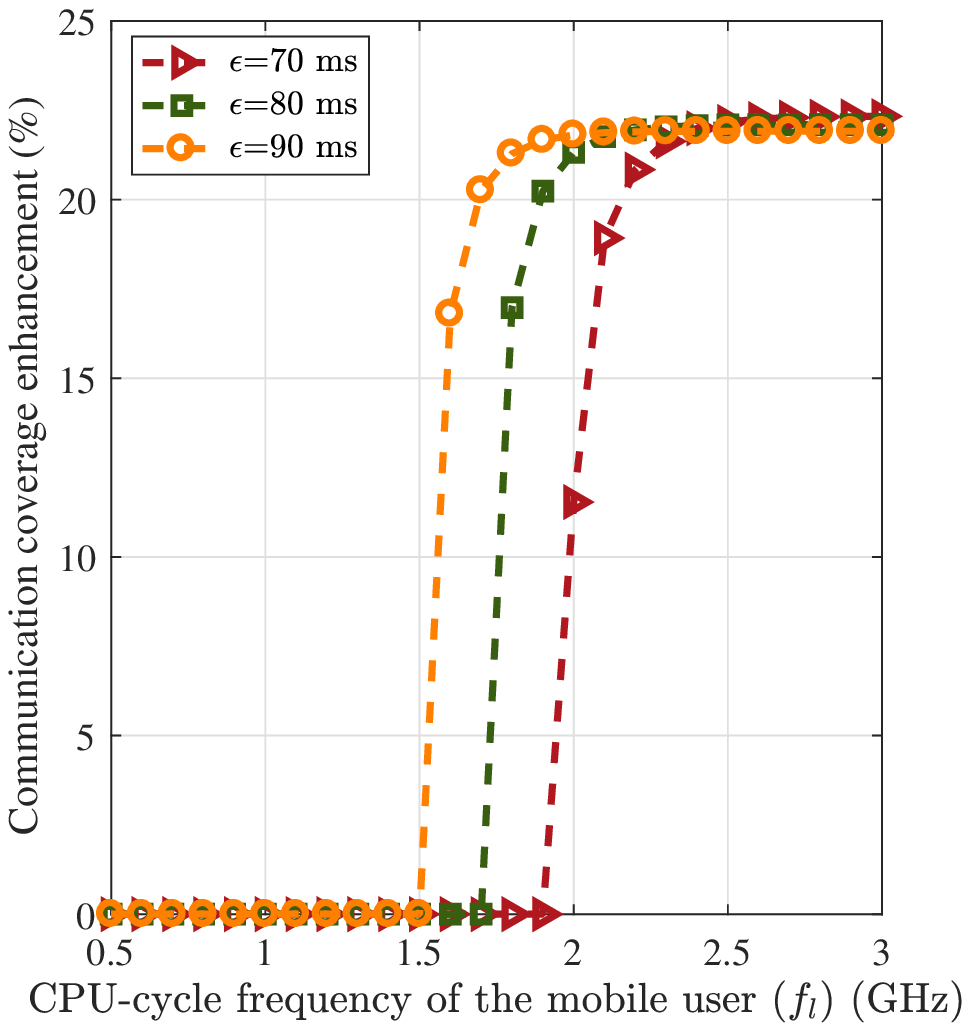}
\caption{\scriptsize{Coverage enhancement ratio.}}
            \label{percent_epsilon}
\end{minipage}

\begin{minipage}[t]{0.213\linewidth}
    \includegraphics[width = 1\linewidth]{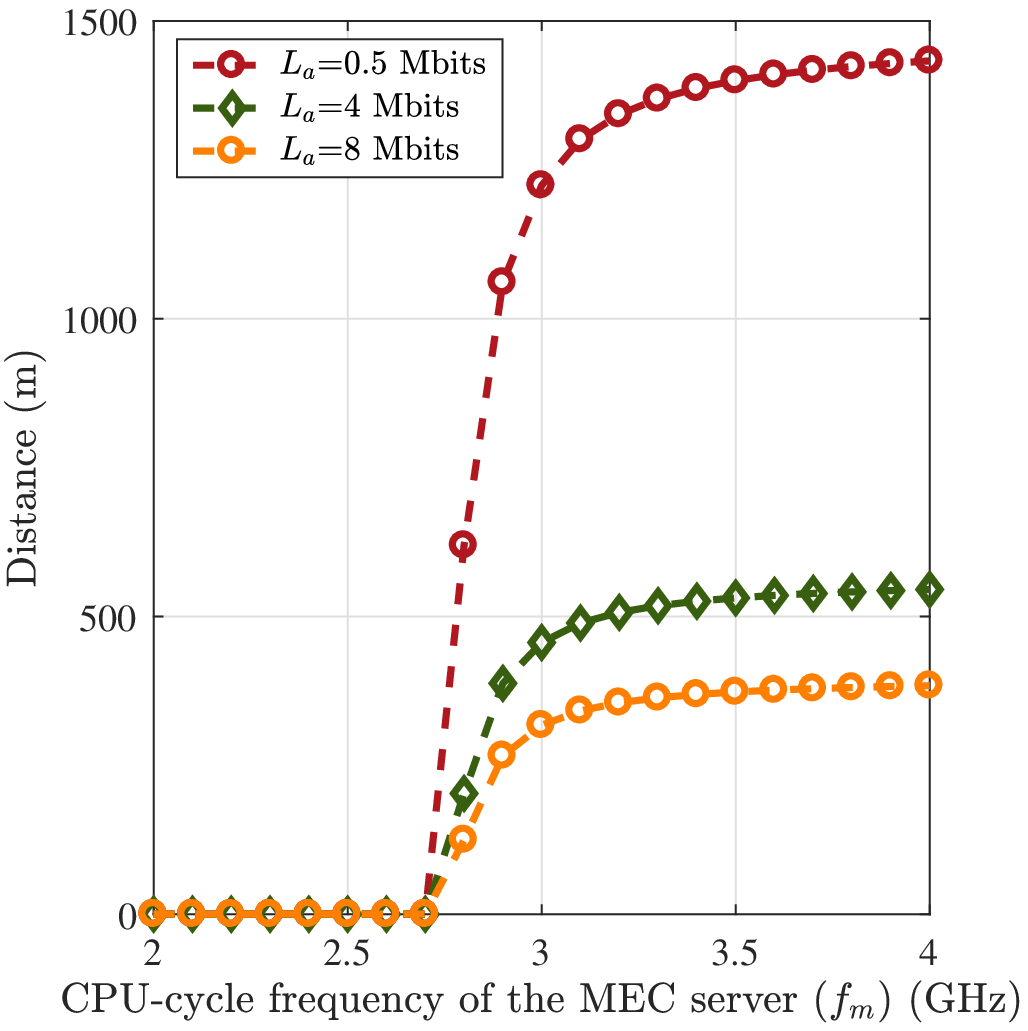}
\caption{\scriptsize{Distance versus the edge server capacity.}}
            \label{sum_fm}
\end{minipage}
\begin{minipage}[t]{0.263\linewidth}
    \includegraphics[width = 1\linewidth]{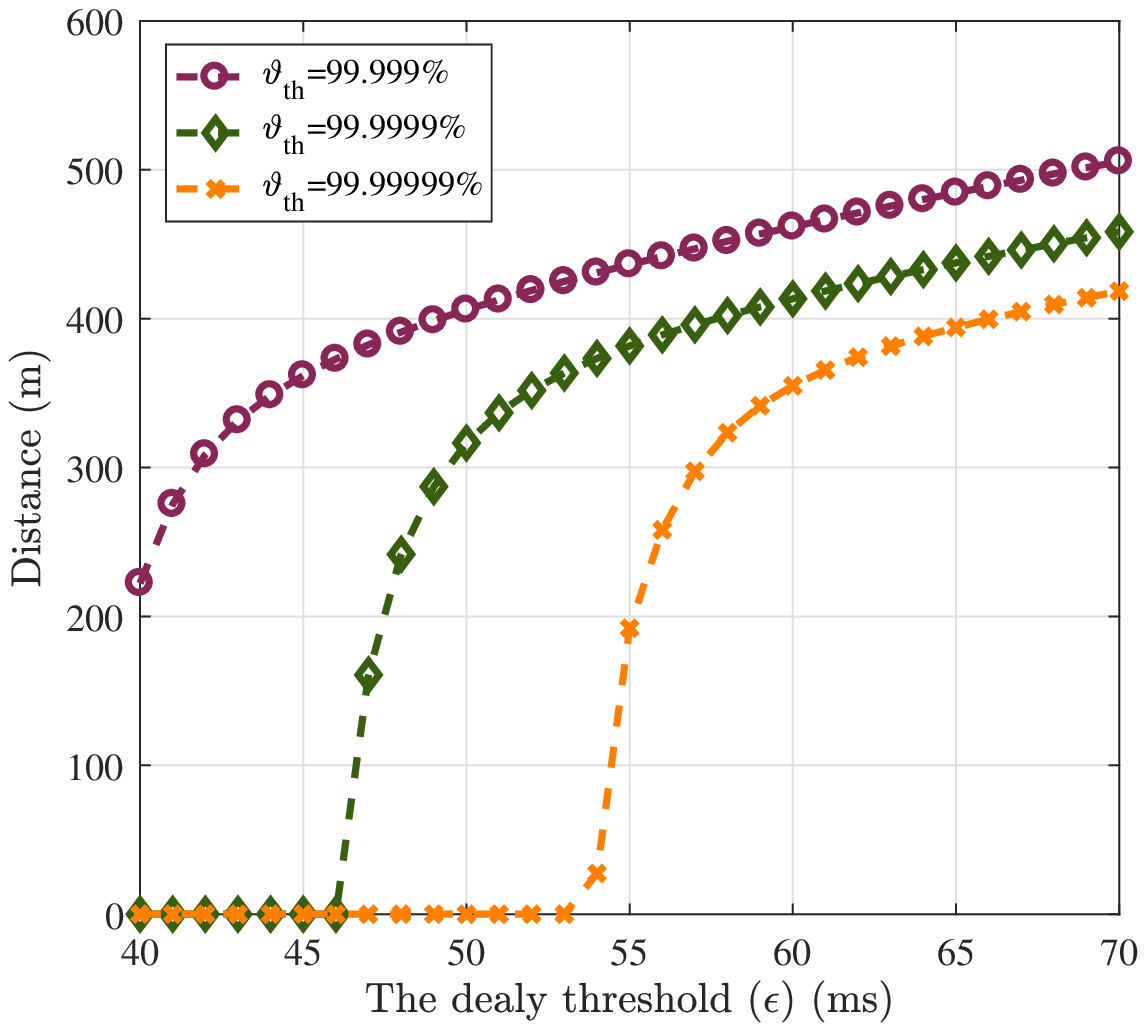}
\caption{\scriptsize{Distance versus the delay threshold.}}
            \label{sum_epsilon}
\end{minipage}
\end{tabular}
\end{figure*}

In general, we can solve (P2) by some classic iterative algorithms, where the disadvantage is that the computational complexity is quite high. As a result, based on the special
structure of the problem, we propose an effective and low-complexity method to solve (P2) by decomposing it into two sub-problems: the data rate threshold minimization problem and the sum communication distance maximization problem.

In this paper, we consider the strategy for all tasks offloading as a benchmark for comparison, where (\ref{equ:1d}) is the prerequisite. Therefore, in our strategy via exploiting mobile computing, (\ref{equ:1d}) is also satisfied, which means all the offloading probabilities are independent of each other. Then, (P2) can thus be described as follows:
\vspace{-0.06cm}
\begin{subequations}
\begin{align}
\mbox{(P3)}\,&\max\limits_{{\boldsymbol{\beta} },{\bf{\tilde f}},{\bf{R}}} {\rm{ }}\sum\nolimits_{k = 1}^K {\frac{{20}}{{\gamma'({\tilde f_k}) \!\cdot\! \ln 10}}W\left\{ {\frac{{\gamma'({\tilde f_k}) \!\cdot\! \ln 10}}{{20{\tilde f_k}}} \!\cdot\! {{10}^{\frac{{\chi ({R_k})}}{{20}}}}} \right\}}   \nonumber \\
&\ {{\rm{s.t.}}}   \, \    \rm(\ref{equ:1b}), \rm(\ref{equ:1e}), \rm(\ref{equ:1f})  \ \rm{and} \  \rm(\ref{equ:3a}),\nonumber
%{R_k}\ge {R_{k,th}},\forall k,\label{equ:4a}\\
%&\qquad\ {\rm{        (1}} - {\beta _k}{\rm{)}}{\lambda _k} < {\mu _{l,k}},\forall k,\label{equ:4b}\\
%&\qquad\ {\rm{        0}} \le {\beta _k} \le {\rm{1, }}\forall k ,\label{equ:4e}\\
%               & \qquad\  {f_k} \in {{\bf{f}}_0},{f_k} \ne {f_j},k \ne j,\forall k,j \label{equ:4f}.
\end{align}
\end{subequations}

It is easy to observe that all transmission rates are also independent of each other due to the aforementioned independence of offloading probabilities. With this, (P3) can be decomposed into the following two optimization sub-problems. First, for each user we optimize its offloading probability to minimize the data rate threshold. Second, based on the obtained minimum data rate threshold of each user, we maximize the sum distance by optimizing frequency allocation.

Specifically, the objective function decreases monotonically as data rate increases according to (\ref{equ:fenxi1}). Therefore, the maximum sum distance can be achieved when the $k$-th user's data rate $R_k$ is equal to the data rate threshold ${R_{k,th}}$ in (\ref{equ:3a}). Moreover, as the data rate threshold is related to offloading probability, we need minimize the data rate threshold to improve the sum distance by optimizing the offloading probability. In what follows, we introduce how to solve the above two sub-problems in detail, respectively.

\subsubsection{Data Rate Threshold Minimization}

\

\noindent  \ \
Based on the above analyses, the first sub-problem that minimizes the data rate threshold ${R_{k,th}}$ by optimizing the offloading probability $\beta_k$ for the $k$-user is given by
\begin{subequations}
\begin{align}
\mbox{(P3-1)}\,&\min\limits_{{\beta_k}} {R_{k,th}},\forall k  \nonumber \\
&\ {{\rm{s.t.}}}  \ {\rm{        (1}} - {\beta _k}{\rm{)}}{\lambda _k} < {\mu _{l,k}}{\rm{, }}\forall k  .\nonumber\\
&\qquad\ {\rm{        0}} \le {\beta _k} \le {\rm{1, }}\forall k .\nonumber
\end{align}
\end{subequations}

Note that ${R_{k,th}}$ is a function only related to $\beta_k$. Thus, we can solve (P3-1) by taking the derivative of ${R_{k,th}}$ with respect to $\beta_k$, thereby obtaining the minimum data rate threshold set ${\bf{R}}_{th}' = \{R_{1,th}',...,R_{k,th}',...,R_{K,th}'\}$ corresponding to the optimal solution set ${\boldsymbol{\beta}}^* = \{\beta_{1}^*,...,\beta_{k}^*,...,\beta_{K}^*\}$.

\subsubsection{Sum Distance Maximization}

\

\noindent  \ \
Based on the results of (P3-1), the second sub-problem for sum distance maximization can be described as:
\begin{subequations}
\begin{align}
\mbox{(P3-2)}\,&\max\limits_{{\bf{\tilde f}},{\bf{R}}} {\rm{ }}\sum\nolimits_{k = 1}^K {\frac{{20}}{{\gamma'({\tilde f_k})\! \cdot\! \ln 10}}W \! \! \left\{\! {\frac{{\gamma'({\tilde f_k}) \! \cdot \! \ln 10}}{{20{\tilde f_k}}} \!\cdot\! {{10}^{\frac{{\chi ({R_k})}}{{20}}}}} \! \right\}}    \nonumber \\
&\ {{\rm{s.t.}}} \quad R_{k} = {{R}}_{k,th}',\forall k  ,\nonumber\\
& \quad\ \ \quad {\tilde f_k} \in {{\bf{\tilde f}}_0},{\tilde f_k} \ne {\tilde f_j},k \ne j,\forall k,j \nonumber.
\end{align}
\end{subequations}
% where ${\bf{R}}_{th}^* = [R_{1,th}^*,...,R_{k,th}^*,...,R_{K,th}^*]$ represents the set of the minimum data rate threshold obtained by (P4-1).

\begin{theorem}
The achievable communication distance (\ref{equ:fenxi1}) of the $k$-th user, $\forall k$,
%, i.e.,
%\begin{align}
%{d_k} = \frac{{20}}{{\gamma'({f_k}) \cdot \ln 10}}W\left\{ {\frac{{\gamma'({f_k}) \cdot \ln 10}}{{20{f_k}}} \cdot {{10}^{\frac{{\chi ({R_k})}}{{20}}}}} \right\} \nonumber
%\end{align}
is a binary strictly increasing function in the frequency range $0.1-0.215$ THz.
Based on this, sort the elements in ${{\bf{\tilde f}}_0}$ and ${\bf{R}}_{th}'$, respectively, as ${\tilde f_0^{1*}}\!\! <\!... \!<\! {\tilde f_0^{k*}} \!\!<\! ...\! <\! \!{\tilde f_0^{K*}}$, $R_{1,{{th}}}^{*} \!\le\! ... \!\le\! R_{k,{{th}}}^{*}\! \le \!... \!\le\! R_{K,{{th}}}^{*}$, then the maximum sum distance of the $K$ users is
\begin{align}
d^*\!\!=\!\!\sum\nolimits_{k = 1}^K \!{\!\frac{{20}}{{\gamma'({{\tilde f_0^{k*}}}) \!\cdot \!\ln \!10}}W\!\!\left\{\! {\frac{{\gamma'({\tilde f_0^{k*}}) \!\cdot \!\ln \!10}}{{20{{\tilde f_0^{k*}}}}} \!\cdot \! {{10}^{\frac{{\chi (R_{k,{{th}}}^{*})}}{{20}}}}} \!\right\}} .\!
\end{align}
\end{theorem}
\begin{proof}
See Appendix A.
\end{proof}
Theorem 1 presents an effective and low-complexity method to maximize sum distance, that is, the smaller the user's minimum rate threshold, the lower the assigned frequency.
%Now, we can obtain the maximum sum distance $d^*$ based on ${{\bf{f}}^*_0} = \{{f_0^{1*}} ,... , {f_0^{k*}} , ... , {f_0^{K*}}\}$ and ${\bf{R}}_{th}^* = \{R_{1,th}^{*},...,R_{k,th}^{*},...,R_{K,th}^{*}\}$.

\vspace{-4pt}
\section{Numerical Results}
\vspace{-4pt}
In this section, we numerically evaluate the performance of the proposed method for improving THz coverage. In Fig.\ref{reliability_THz}, Fig.\ref{reliability_4G_1} and Fig.\ref{reliability_4G_2}, we first verify the effectiveness of our framework for providing 6G URLLC services when E2E delay threshold is $80$ ms. Subsequently, we validate the distance improvement for the multi-user case in Fig.\ref{dis_fl_epsilon}, Fig.\ref{percent_epsilon}, Fig.\ref{sum_fm} and Fig.\ref{sum_epsilon}. Table II gives a summary of the simulation parameters.
\vspace{-12pt}
\begin{table}[htbp]
\setlength{\abovecaptionskip}{-0.05cm}
\caption{Simulation Parameters}
\small
\centering
%\resizebox{\textwidth}{1mm}{
\begin{tabular}{m{1cm}m{1.4cm}|m{0.9cm}m{2cm}}
\toprule
Parameter    & Value                       & Parameter    & Value                          \\[-2.5pt]
\midrule
$B$        & $10$ GHz                       &  $p$        & $100$ mW \\
$N_0$        & $-40$ dBm                     & $\tilde f$        & $[0.1,0.2]$ THz                          \\
$G_t$        & $20$ dBi                     & $L_a$        & $8$ Mbits \cite{La}                 \\
$G_r$        & $20$ dBi                    & $\mu_a$        & $10^7$ cycles\cite{mu}                         \\
\bottomrule
\end{tabular}
\end{table}
\vspace{-5pt}

Fig.\ref{reliability_THz}, Fig.\ref{reliability_4G_1} and Fig.\ref{reliability_4G_2} show the reliability in high- and low-band communication, respectively. We can see that the reliability of all jobs only executing in edge server is lower than that via exploiting mobile computing. The reason is that the transmission delay will be reduced when mobile user executes some tasks, which reduces total delay and further improve reliability. In addition, Fig.\ref{reliability_4G_1} and Fig.\ref{reliability_4G_2} also indicate that the reliability in low-frequency communication is too low to satisfy the reliability threshold (${\vartheta _{th}}{\rm{ = }}99.999{\rm{\% }}$). The reason is that the data rate in low-frequency communication is relatively low, which leads to a large transmission delay. Accordingly, the reliability of the edge computing is low, further the whole reliability is also low.

In Fig.\ref{dis_fl_epsilon} and Fig.\ref{percent_epsilon}, we compare the coverage performance in the single user case for the two different schemes, i.e., only edge computing, and exploiting mobile computing, where the reliability threshold is ${\vartheta _{th}}{\rm{ = }}99.999{\rm{\% }}$. It can be seen that the communication distances are greatly improved when we exploit mobile computing. Specifically, when the user's computation capacity is very limited, it cannot execute tasks locally, and all tasks are processed in edge server. At this time, the communication distances for these two schemes are the same. However, when user's capacity become larger, it can process partial tasks even under a high reliability requirement, thus the burden of wireless communication link is reduced, thereby increasing the communication coverage.

Fig.\ref{sum_fm} illustrates the relationships between distance and edge server capacity with different $L_a$, where the number of mobile users is $10$. It can be observed that the distance increases with respect to the edge server capacity while decreases with respect to the data size. This can be explained as, under the specific reliability threshold and delay threshold, when the edge server capacity increases, the data rate can be lower, which results in the improvement of the distance. On the contrary, when the data size increases, the transmission rate should be increased to satisfy the same reliability requirement. Therefore, the maximum sum distance will decrease. Besides, we can see that the distance could stay stationary when the edge server capacity is sufficiently large. This is expected since the fact that the edge server with large capacity offers a very low computation execution time, and there is no need to increase data rate to meet reliability requirements.

In Fig.\ref{sum_epsilon}, the distance versus the delay threshold is plotted with different reliability thresholds. We also assume the number of users is $10$. It can be seen that the distance increases as the delay threshold increases, and decreases with respect to the reliability threshold, respectively. This is expected since the fact that with the increase of delay threshold and the decrease of reliability threshold, respectively, all tasks can be executed within a larger delay, which leads to the transmission rate requirement decreasing. Further, according to (\ref{equ:rate}), the achievable distance becomes larger. Additionally, it is worth noting that the distance is $0$ when the reliability threshold is largely high (e.g., $99.99999{\rm{\% }}$), and the delay threshold is sufficiently low (e.g., $50$ ms). The reason is that the above reliability and delay requirements are sufficiently strict, all tasks cannot be computed in time under the specific local and edge computation capacities.

\section{Conclusions}
In this paper, we aimed to improve THz coverage aided by mobile computing based on the joint design of communication and computing. We established an optimization problem which maximized the communication distance by optimizing offloading probability and frequency allocation. This problem was non-convex, followed by we proposed an effective and low-complexity method to solve it. Finally, numerical results demonstrated the effectiveness of our schemes.
    \begin{appendices}
\section{  }
\textbf{\emph{On the one hand}}, we can prove ${d_k}$ in (\ref{equ:fenxi1}) is a binary strictly increasing function in the frequency range $0.1-0.215$ THz by the following two steps:

1) By calculating the first- and second-order partial derivatives for  (\ref{equ:fenxi1}), respectively, we obtain $\frac{{\partial \xi ({R_k},{\tilde f_k})}}{{\partial {\tilde f_k}}}{\rm{ = }}\frac{{{\partial ^2}{d_k}}}{{\partial {R_k}\partial {\tilde f_k}}} > 0$ in $0.1-0.215$ THz .

2) Based on 1), we further obtain ${d_k}({R_k} + \Delta {R_k},{\tilde f_k} + \Delta {\tilde f_k}) + {d_k}({R_k},{\tilde f_k}) - {d_k}({R_k} + \Delta {R_k},{\tilde f_k}) - {d_k}({R_k},{\tilde f_k} + \Delta {\tilde f_k}) > 0$.

Specifically, letting ${y = }\frac{{\gamma'\ln 10}}{{20}} \cdot {\rm{1}}{{\rm{0}}^{\frac{\chi }{{20}}}} \cdot \tilde f_k^{ - 1}$, we have
\begin{align}
\xi (\!{R_k},{\tilde f_k}\!)\!=\! \frac{{\partial {d_k}}}{{\partial {R_k}}}\! =\!  - \frac{{{2^{\frac{{{R_k}}}{B_k}}} \!\cdot\! 10\ln 2}}{{{B_k}({2^{\frac{{{R_k}}}{B_k}}} \!-\! 1) \!\cdot\! \ln 10}} \!\cdot\! \frac{{{\rm{W(}}y)}}{{\gamma'{\rm{(W(}}y) \!+\! 1)}},
\end{align}

\begin{align}
\frac{{\partial \xi ({R_k},{\tilde f_k})}}{{\partial {\tilde f_k}}}= &\frac{{{\partial ^2}{d_k}}}{{\partial {R_k}\partial {\tilde f_k}}}\notag\\
=&  - \frac{{{2^{\frac{{{R_k}}}{B_k}}} \cdot 10\ln 2}}{{B_k({2^{\frac{{{R_k}}}{B_k}}} - 1) \cdot \ln 10}}\notag\\
& \!\cdot\! \frac{{\frac{{\partial {\rm{W(}}y)}}{{\partial {\tilde f_k}}} \cdot \gamma' \!-\! {\rm{W(}}y) \!\cdot\! \frac{{\partial \gamma'}}{{\partial {\tilde f_k}}}\! \cdot\! {\rm{(W(}}y) \!+\! 1)}}{{{{[\gamma'{\rm{(W(}}y) \!+\! 1)]}^2}}}.
\end{align}
Here, $\frac{{\partial \gamma'}}{{\partial {\tilde f_k}}}$ and $\frac{{\partial {\rm{W(}}y)}}{{\partial {\tilde f_k}}}$ are the derivations of $\gamma'$ and ${\rm{W(}}y)$ with respect to ${\tilde f_k}$, respectively, which can be given by
\begin{align}
\frac{{\partial \gamma'}}{{\partial {\tilde f_k}}}{\rm{ = }}\sum\limits_{i = 1}^7 {2{a_i}\frac{{{b_i} - {\tilde f_k}}}{{c_i^2}}{e^{ - {{(\frac{{{\tilde f_k} - {b_i}}}{{{c_i}}})}^2}}}}  > 0,
\end{align}
\begin{align}
\frac{{\partial {\rm{W(}}y)}}{{\partial {\tilde f_k}}}{\rm{ = }}\frac{{\frac{{\ln 10}}{{20}} \cdot {\rm{1}}{{\rm{0}}^{\frac{\chi }{{20}}}} \cdot {\rm{W(}}y)}}{{y{\rm{(W(}}y) + 1)}} \cdot \frac{{{\tilde f_k} \cdot \frac{{\partial \gamma'}}{{\partial {\tilde f_k}}} - \gamma'}}{{\tilde f_k^2}}.
\end{align}
Furthermore, we can obtain
\begin{align}
\left\{ {\begin{array}{*{20}{c}}
{\frac{{{\tilde f_k} \cdot \frac{{\partial \gamma'}}{{\partial {\tilde f_k}}} - \gamma'}}{{\tilde f_k^2}} > 0,\tilde f_k > 216.5692 \ {\rm{ GHz }}}\\
{{\rm{ }}\frac{{{\tilde f_k} \cdot \frac{{\partial \gamma'}}{{\partial {\tilde f_k}}} - \gamma'}}{{\tilde f_k^2}} \le 0,\tilde f_k \le 216.5692 \ {\rm{ GHz}}}
\end{array}} \right..
\end{align}
Therefore, when $\tilde f_k \le 216.5692$ GHz, we have $\frac{{\partial {\rm{W(}}y)}}{{\partial {\tilde f_k}}} \le 0$, and further obtain $\frac{{\partial \xi ({R_k},{\tilde f_k})}}{{\partial {\tilde f_k}}}{\rm{ = }}\frac{{{\partial ^2}{d_k}}}{{\partial {R_k}\partial {\tilde f_k}}} > 0$. This completes the proof of 1).

Next, we will prove 2). Based on the results of 1), we can arrive at $\xi ({R_k},{\tilde f_k} + \Delta {\tilde f_k}) - \xi ({R_k}, {\tilde f_k}) > 0$, where $\Delta {\tilde f_k} > 0$.
In addition, we assume $\rho ({R_k},{\tilde f_k}){\rm{ = }}{d_k}({R_k},{\tilde f_k} + \Delta {\tilde f_k}) - {d_k}({R_k},{\tilde f_k})$. By taking the derivative of ${\rho ({R_k},{\tilde f_k})}$ with respect to $R_k$, we can arrive at
\begin{align}
\frac{{\partial \rho ({R_k},{\tilde f_k})}}{{\partial {R_k}}} =&\frac{{\partial {d_k}({R_k} ,{\tilde f_k}+ \Delta {\tilde f_k})}}{{\partial {R_k}}} - \frac{{\partial {d_k}({R_k},{\tilde f_k})}}{{\partial {R_k}}}\notag\\
 = &\xi  ({R_k},{\tilde f_k} + \Delta {\tilde f_k}) - \xi  ({R_k},{\tilde f_k}) > 0,
\end{align}
which leads to $\rho ({R_k} + \Delta {R_k},{\tilde f_k}) - \rho ({R_k},{\tilde f_k}) > 0$. Then we further obtain ${d_k}({R_k} + \Delta {R_k},{\tilde f_k} + \Delta {\tilde f_k}) + {d_k}({R_k},{\tilde f_k}) - {d_k}({R_k} + \Delta {R_k},{\tilde f_k}) - {d_k}({R_k},{\tilde f_k} + \Delta {\tilde f_k}) > 0$. According to the definition of the binary strictly increasing function, we can infer that the distance of the $k$-th user, i.e., ${d_k}({R_k},{\tilde f_k}) = \frac{{20}}{{\gamma({\tilde f_k}) \cdot \ln 10}}W\left\{ {\frac{{\gamma({\tilde f_k}) \cdot \ln 10}}{{20{\tilde f_k}}} \cdot {{10}^{\frac{{\chi ({R_k})}}{{20}}}}} \right\}$ is a binary strictly increasing function when the frequency $\tilde f_k$ belongs to $0.1-0.215$ THz.

\textbf{\emph{On the other hand}}, sort the elements in ${{\bf{\tilde f}}_0}$ and ${\bf{R}}_{th}'$, respectively, as ${\tilde f_0^{1*}}\! <\!... \!<\! {\tilde f_0^{k*}} \!<\! ...\! <\! {\tilde f_0^{K*}}$, $R_{1,{{th}}}^{*} \!\le\! ... \!\le\! R_{k,{{th}}}^{*}\! \le \!... \!\le\! R_{K,{{th}}}^{*}$, and define ${\bf{\tilde f}}_0'$ is a permutation of ${\bf{\tilde f}_0}$.  The rest of Theorem 1 can be proved when the following inequality is always satisfied, i.e.,
\begin{align}
\sum\limits_{i = 1}^K\! {d(R_{i,th}^{*},{\tilde f_0^{i*}})\! \ge \!} \sum\limits_{i = 1}^K \!{d(R_{i,th}^{*},{\tilde f_0^{i'}})\! \ge\! } \sum\limits_{i = 1}^K \! {d(R_{i,th}^{*},{\tilde f_0^{K - i + 1*}})}.
\label{proof}
\end{align}
Fortunately, we can use the inductive reasoning to complete the proofs of the left inequality, i.e., $\sum\nolimits_{i = 1}^K \! {d(\!R_{i,th}^{*},\!{\tilde f_0^{i*}}\!) \!\ge } \!\sum\nolimits_{i = 1}^K\! {d(\!R_{i,th}^{*},\!{\tilde f_0^{i'}}\!)}$, and the right inequality, i.e., $\sum\nolimits_{i = 1}^K\! {d(\!R_{i,th}^{*},\!{\tilde f_0^{i'}}\!)\! \ge \!} \sum\nolimits_{i = 1}^K \!{d(\!R_{i,th}^{*},\!{\tilde f_0^{K - i + 1*}}\!)}$, respectively.

Specifically, we use the inductive reasoning to complete the proof of the left inequality by the following three steps:

1) $K = 1$, the inequality transforms into an equation, which is obviously true.

2) Assuming that when $K = n$, the inequality holds true.

3) In what follows, we prove that the inequality also holds true when $K = n+1$.
%\begin{itemize}
%  \item When $K = 1$, the inequality transforms into an equation, thus the above equation is obviously true.
%\end{itemize}
%
%\begin{itemize}
%  \item We assume that when $K = n$, the inequality holds true.
%\end{itemize}
%\begin{itemize}
%  \item With this, we need prove that the inequality also holds true when $K = n+1$.
%\end{itemize}

First, the data from the above two permutations ( i.e., ${\bf{\tilde f_0}}'$ and $ {\bf{R}}_{th}^* $ ), respectively, are combined in pairs, and they form $n+1$ pairs, i.e., $(R_{{i_1},th}^{*},{\tilde f_0^{{i_1}'}})$, $(R_{{i_2},th}^{*},{\tilde f_0^{{i_2}'}})$ ,..., $(R_{{i_{n + 1}},th}^{*},{\tilde f_0^{{i_{n + 1}}'}})$. Then arrange $R_{{i_j},th}^{*} $ and ${\tilde f_0^{{i_j}'}}(\forall j = 1,2,...n + 1)$, respectively, in order. If these two data are in the same position in their respective arrangements, $(R_{{i_j},th}^{*},{\tilde f_0^{{i_j}'}})$ is called same-order pair.
On the contrary, if the sum of their position indexes in their respective arrangements is $n+1$, $(R_{{i_j},th}^{*},{\tilde f_0^{{i_j}'}})$ is called reverse-order pair.

\begin{itemize}
  \item If there is one same-order pair $(R_{{i_l},th}^{*},{\tilde f_0^{{i_l}'}})$ in $n+1$ pairs, according to the inductive reasoning, we can arrive at
  \begin{align}
\sum\limits_{j = 1,j \ne l}^{n + 1} {d(R_{{i_j},th}^{*},{\tilde f_0^{{i_j}*}}) \ge }& \sum\limits_{j = 1, j \ne l}^{n + 1} {d(R_{{i_j},th}^{*},{\tilde f_0^{{i_j}'}})}.
\end{align}
Adding $d(R_{{i_l},th}^{*},{\tilde f_0^{{i_l}*}})$ to both sides of the above inequality, we have
\begin{align}
\sum\limits_{j = 1}^{n + 1} {d(R_{{i_j},th}^{*},{\tilde f_0^{{i_j}*}}) \ge }& \sum\limits_{j = 1}^{n + 1} {d(R_{{i_j},th}^{*},{\tilde  f_0^{{i_j}'}})}\label{equ:before}.
\end{align}
\item If there is not even one same-order pair in $n+1$ pairs, we assume that $(R_{{i_{n + 1}},th}^{*},{\tilde f_0^{{i_r}'}})$ is a same-order pair, and $R_{{{i_1},{th}}}^{*} \le R_{{{i_2},{th}}}^{*} \le ... \le R_{{{i_{n + 1}},{th}}}^{*}$. We have ${\tilde f_0^{{i_r}'}} \ge {\tilde f_0^{{i_{n + 1}}'}}$ as well as $R_{{{i_r},{th}}}^{*} \le R_{{{i_{n + 1}},{th}}}^{*}$. After exchanging ${\tilde f_0^{{i_r}'}}$ and ${\tilde f_0^{{i_{n + 1}}'}}$, we compute the sum of distance, i.e., ${D_1}$, which can be given by
  \begin{align}
{D_1} =&\sum \limits_{j = 1, j \ne r,j \ne n+1}^{n + 1} {d(R_{{i_j},th}^{*},{\tilde f_0^{{i_j}'}})}  + d(R_{{i_r},th}^{*},{\tilde f_0^{{i_{n + 1}'}}})\notag\\
 &+ d(R_{{i_{n + 1}},th}^{*},{\tilde f_0^{{i_r}'}}).
\end{align}
Letting ${D_2} = \sum\limits_{j = 1}^{n + 1} {d(R_{{i_j},th}^{*},{\tilde f_0^{{i_j}'}})}$, we compute ${D_2} - {D_1}$ as follows,
  \begin{align}
\!\!{D_2} \!-\!\! {D_1}&\!=\!\! \sum\limits_{j = 1}^{n + 1}\! {d(R_{{i_j},th}^{*},{\tilde f_0^{{i_j}'}}\!)} \! -\!\! \sum \limits_{j = 1, j \ne r,j \ne n+1}^{n + 1} \!\!\!\!{d(R_{{i_j},th}^{*},{\tilde  f_0^{{i_j}'}}\!)}\notag\\
&\quad  - d(R_{{i_r},th}^{*},{\tilde f_0^{{i_{n + 1}'}}})
- d(R_{{i_{n + 1}},th}^{*},{\tilde f_0^{{i_r}'}})\notag\\
&=\!d(R_{{i_r},th}^{*},{\tilde f_0^{{i_r}'}}\!) + d(R_{{i_{n + 1}},th}^{*},{\tilde f_0^{{i_{n + 1}'}}}\!) \notag \\
&\quad- d(R_{{i_r},th}^{*},{\tilde f_0^{{i_{n + 1}'}}}\!) - d(R_{{i_{n + 1}},th}^{*},{\tilde f_0^{{i_r}'}}\!).
\end{align}
Furthermore, as ${d_k}({R_k},{\tilde f_k})$ is a binary strictly increasing function, we can obtain ${D_2} - {D_1} \le 0$.
Moreover, since there is one same-order pair $(R_{{i_{n + 1}},th}^{*},{\tilde f_0^{{i_r}'}})$ in ${D_1}$, according to (\ref{equ:before}) we have
  \begin{align}
\sum\limits_{j = 1}^{n + 1} {d(R_{{i_j},th}^{*},{\tilde f_0^{{i_j}*}}) \ge } {D_1} \ge {D_2},
\end{align}
thereby obtaining
  \begin{align}
\sum\limits_{j = 1}^{n + 1} {d(R_{{i_j},th}^{*},{\tilde f_0^{{i_j}*}}) \ge } \sum\limits_{j = 1}^{n + 1} {d(R_{{i_j},th}^{*},{\tilde f_0^{{i_j}'}})} .
\end{align}
\end{itemize}

Now, we complete the proof of the left inequality in (\ref{proof}). Analogously, the proof of the right inequality can be completed  in the same way from the perspective of the reverse-order pair.

After completing all the proofs of (\ref{proof}), it is easy to observe that the maximum sum distances is given by
\begin{align}
d^*\!\!=\!\!\sum\limits_{k = 1}^K \!{\!\frac{{20}}{{\gamma'({{\tilde f_0^{k*}}}) \!\cdot \!\ln \!10}}W\!\!\left\{\! {\frac{{\gamma'({\tilde f_0^{k*}}) \!\cdot \!\ln \!10}}{{20{{\tilde f_0^{k*}}}}} \!\cdot \! {{10}^{\frac{{\chi (R_{k,{{th}}}^{*})}}{{20}}}}} \!\right\}} .\!
\end{align}
\end{appendices}

\end{document}